\documentclass[letterpaper, 10 pt, conference]{ieeeconf}  

\IEEEoverridecommandlockouts                              

\usepackage{amsmath,amssymb,amsfonts}
\usepackage{stackengine}
\usepackage{mathrsfs} 
\usepackage{mathtools}
\usepackage{cite}
\usepackage{algorithmic}
\usepackage{graphicx}
\usepackage{textcomp}
\usepackage{siunitx}
\usepackage{pgfplots} 
\pgfplotsset{compat=newest}
\usepackage{tikz}
\usepackage{svg}
\usepackage[acronym]{glossaries} 
\usepackage[hidelinks, raiselinks=true, bookmarksnumbered=true, pdfusetitle, hyperfootnotes=false]{hyperref}
\usepackage[capitalise]{cleveref}
\crefname{equation}{}{} 
\def\BibTeX{{\rm B\kern-.05em{\sc i\kern-.025em b}\kern-.08em
    T\kern-.1667em\lower.7ex\hbox{E}\kern-.125emX}}

\newtheorem{theorem}{Theorem}

\newtheorem{assumption}{Assumption}

\newcommand{\ubar}[1]{\stackunder[1.2pt]{$#1$}{\rule{.8ex}{.075ex}}}

\newacronym{MPC}{MPC}{model predictive control}
\newacronym{NN}{NN}{neural network}
\newacronym{LUT}{LUT}{look-up table}

\tikzstyle{block} = [draw, rectangle,
minimum height=1cm, minimum width=3.5cm, text width=3.5cm, align=center]
\tikzstyle{sum} = [draw, fill=blue!20, circle, node distance=1cm]
\tikzstyle{input} = [coordinate]
\tikzstyle{output} = [coordinate]
\tikzstyle{pinstyle} = [pin edge={to-,thin,black}]

\title{\LARGE \bf
Contract-based hierarchical control using \\ predictive feasibility value functions
}

\author{Felix Berkel$^{1}$, Kim Peter Wabersich$^{1}$, Hongxi Xiang$^{2}$, Elias Milios$^{1}$
\thanks{$^{1}$Elias Milios, Kim P. Wabersich, and Felix Berkel are with the Corporate Research of Robert Bosch GmbH, 71272 Renningen, Germany.
        Email: \{Elias.Milios, KimPeter.Wabersich, Felix.Berkel\}@de.bosch.com}%
\thanks{$^{2}$Hongxi Xiang is a student of ETH Zurich, Switzerland.
Email: xiang-hongxi@outlook.com}%
}

\begin{document}

\maketitle
\thispagestyle{empty}
\pagestyle{empty}

\begin{abstract}
    Today's control systems are often characterized by modularity and safety requirements to handle complexity, resulting in the use of hierarchical control structures. Although hierarchical model predictive control offers favorable properties, achieving a provably safe, yet modular design remains a challenge. 
    This paper introduces a contract-based hierarchical control strategy to improve the performance of control systems facing challenges related to model inconsistency and independent controller design across hierarchies. 
    We consider a setup where a higher-level controller generates references that affect the constraints of a lower-level controller, which is based on a soft-constrained MPC formulation.
   The optimal slack variables of the lower-level MPC serve as the basis for a contract that allows the higher-level controller to assess the feasibility of the reference trajectory without exact knowledge of the model, constraints, and cost of the lower-level controller. 
    To ensure computational efficiency while maintaining model confidentiality, we propose using an explicit function approximation, such as a neural network, to represent the cost of optimal slack values.  
    The approach is tested for a hierarchical control setup consisting of a planner and a motion controller as commonly found in autonomous driving.
\end{abstract}

\section{Introduction}
\label{sec:introduction}
%
Hierarchical control structures decompose complex control problems into smaller, manageable sub-problems, each handled by a dedicated controller. 
This modularization is crucial for addressing the inherent complexity and stringent safety requirements of modern dynamical systems, simplifying development and improving robustness. 
In this context, hierarchical \gls{MPC} has emerged as a prominent approach. \gls{MPC}, an advanced control method, predicts future system behavior using a model and optimizes control inputs over a finite horizon, often under constraints. 
Hierarchical \gls{MPC} extends this by coordinating multiple \gls{MPC} controllers across different time scales, integrating decision-making across various abstraction levels while ensuring critical safety. 
This is particularly relevant for systems with multi-rate dynamics or plant-wide optimization. 
Its diverse applications include process control \cite{backx2000integration}, water networks \cite{ocampo2012hierarchical}, power grids \cite{berkel2013load}, and autonomous vehicles \cite{kogel2023safe}. For a comprehensive survey, see \cite{scattolini2009architectures}.

One challenge with hierarchical \gls{MPC} is the use of different models at various layers of the control architecture~\cite{scattolini2009architectures}. 
At the higher levels, simpler models are utilized, which are suitable for capturing the slow system dynamics or for facilitating optimization over a long planning horizon. 
In contrast, at the lower levels, more detailed models are implemented to accurately represent the fast dynamics of the system. 
This layered approach can lead to significant discrepancies between the planned behavior at higher levels and the actual behavior of the system, potentially resulting in suboptimal performance and safety risks, as noted in~\cite{barcelli2010hierarchical}.
Another challenge arises from the independent design of higher-level and lower-level controllers, often carried out by different teams or even different companies, which is, for instance, common in the automotive industry. 
This independence can lead to significant integration effort and may even prohibit the combination of certain components, as highlighted in~\cite{bathge2018contract}. 
The situation is further complicated when knowledge exchange regarding models may be restricted due to intellectual property restrictions. 

\par
\paragraph*{Contributions}
This paper presents a contract-based hierarchical control strategy to address the challenges of modularization and safety. 
At the higher level, a controller generates a reference sequence for the lower-level controller to track, which also affects the constraints for the lower-level controller. The lower-level controller is designed as a soft-constrained \gls{MPC}, such that the corresponding optimal slack variables indicate feasibility of a given state measurement and reference sequence.
As a result, the higher-level controller can assess feasibility of the reference trajectory for the underlying control problem through the optimal slack variables directly, even in the presence of model discrepancies.
Specifically, we formulate an optimization problem to determine the optimal slack variables for a given state and reference sequence. 
We refer to its optimal value as the predictive feasibility value function.

Since the predictive feasibility value function is implicitly defined and requires solving an optimization problem for evaluation, we propose an efficient approximation using an explicit representation, such as a \gls{NN}. 
This approach enables the higher-level controller to efficiently evaluate feasibility, rendering the optimization problem computationally tractable for real-time implementation. 
Furthermore, this approximation allows the higher-level controller to operate without access to detailed information about the lower-level optimization problem. Instead, it can rely on an abstract description, the approximation of the predictive feasibility value function. 
Notably, while predictive value functions and their approximation are a known concept in control engineering (see, e.g., \cite{chatzikiriakos2024learning}), their application as a contract in hierarchical control, as presented here, is novel. 
This effectively addresses the independent design problem, as the lower-level controller may prefer not to disclose its optimization problem due to intellectual property concerns. 
Notably, this method avoids the shortcomings of existing contract-based approaches as outlined in the next paragraph.

\paragraph*{Related Work}
Hierarchical control has evolved significantly. Early approaches, like those in \cite{scattolini2007hierarchical, picasso2010mpc}, focused on robust higher-level controllers accommodating lower-level decisions. 
More recent mission-based hierarchical MPC frameworks, e.g., \cite{koeln2018two, koeln2020vertical}, prioritize recursive feasibility over stability. 
A commonality is that both levels use the same model, ensuring prediction consistency. 
However, in applications like autonomous driving, higher-level planners often cannot accommodate complex lower-level vehicle models, limiting applicability. 
\cite{barcelli2010hierarchical, bathge2018contract} explore contract-based designs where a higher-level MPC and a lower-level linear controller operate with the lower-level model undisclosed to the higher level. 
These methods, however, have several shortcomings: they are restricted to linear systems and controllers, which limits performance for constrained systems, introduce conservatism through robust design, and require bounded changes in upper-level reference signals. 
These factors render them unsuitable for complex applications like planning and motion control in autonomous driving.
\section{Problem formulation}\label{sec:problem_formulation}
\begin{figure}[t!]
    \centering
    \begin{tikzpicture}[scale = 1]
        \node [block, minimum height=1.35cm, align=center] at (0cm, 0cm) (planner) {higher-level controller\\ \eqref{eq:opt_planning}/\eqref{eq:opt_planning_contract_based}};
        \node [block, minimum height=1.35cm, align=center] at (0cm, -2cm) (lowlevelcontrol) {lower-level controller\\
        \eqref{eq:soft_constrained_mpc}};
        \node [block, minimum height=1.35cm] at (0cm, -4cm) (system) {system \eqref{eq:system}};
        \draw [->,very thick]  (system.east) -- +(0.5cm,0cm) |- (lowlevelcontrol.east);
        \draw [->,very thick]  (system.east) -- +(0.5cm,0) |- (planner.east);
        \draw [->,very thick]  (planner.south) -- (lowlevelcontrol.north);
        \draw [->,very thick]  (lowlevelcontrol.south) -- (system.north);
        \draw [->, very thick, gray, dashed] (lowlevelcontrol.west) -- +(-0.5cm, 0) |- (planner.west);
        \node at (3.5cm, -2.0cm) {$x(k)$, $x^{\mathrm{H}}(k)$};
        \node at (-2.75cm, -1.0cm) {\textcolor{gray}{$h_{\mathrm{C}}$}};
        \node at (0.5cm, -3.0cm) {$u(k)$};
        \node at (0.5cm, -1.0cm) {$r^{\mathrm{H,*}}_{\cdot|k_{\mathrm{H}}}$};
    \end{tikzpicture}
    \caption[Planner]{Considered controller architecture: The higher-level controller generates references  $r^{\mathrm{H,*}}_{\cdot|k_{\mathrm{H}}}$ which are tracked by the lower-level controller that applies inputs $u(k)$ to the system. 
    The gray part represents a contract designed offline (before operation), allowing the higher-level controller to assess feasibility of a given trajectory for the lower-level controller during operation. 
    $x(k)$ and $x^{\mathrm{H}}(k)$ are the states of the lower- and higher-level controller, respectively.}
    \label{fig:ADarchitecture}
\end{figure}
We consider nonlinear discrete-time systems of the form
\begin{equation}
    \begin{aligned}
    \label{eq:system}
    x(k+1) &= f(x(k),u(k)),
\end{aligned}
\end{equation}
where $x(k)\in\mathbb R^{n_x}$ represents the system state, $u(k)\in\mathbb R^{n_u}$ denotes the control input, and $k \in \mathbb{Z}_{\mathrm{\geq0}}$ is the sampling instant.
The discrete-time system is derived from a continuous-time system using a sampling time $T_{\mathrm{L}} \in \mathbb{R}_{> 0}$.
System~\cref{eq:system} is subject to state and input constraints of the form
\begin{align}\label{eq:constraints}
    x(k) \in \mathcal X, \qquad 
    u(k) \in \mathcal U,
\end{align}
where  $\mathcal X:= \{x|c_{\mathrm{x}}(x) \leq 0 \} \subseteq \mathbb R^{n_x}$ and $\mathcal U:=\{u|c_{\mathrm{u}}(u) \leq 0 \} \subset \mathbb R^{n_u}$. 
%

The considered controller architecture is hierarchical, consisting of a higher-level controller that provides reference trajectories to a lower-level controller, which then applies the inputs to the system, see Fig.~\ref{fig:ADarchitecture}.
The lower-level controller operates with sampling time $T_{\mathrm{L}}$, while the higher-level controller is sampled with $T_{\mathrm{H}} := N_{\mathrm{L}} \cdot T_{\mathrm{L}}$ with $N_{\mathrm{L}} \in \mathbb{Z}_{\geq 1}$. 
The primary objective of the higher-level controller is to optimize the system's performance over longer time horizons, whereas the lower-level controller is responsible for compensating for fast disturbances. 
This hierarchical structure is commonly found in various applications where the overall task of the controller is distributed between the higher and lower levels, see for instance~\cite{scattolini2009architectures}. 

We first consider a mission-based scenario, where the system operates for a finite duration starting from $k=0$ and ending at $N_{\mathrm{H}}T_{\mathrm{H}}$ with mission horizon $N_{\mathrm{H}} \in \mathbb{Z}_{\geq 1}$ and the higher-level controller plans for the entire mission. 
It computes a reference trajectory $r^{\mathrm{H}}\in\mathbb R^{n_{\mathrm{Hr}}}$, which is subsequently passed to the lower-level controller. 
The reference is an input to the nonlinear discrete-time model 
\begin{equation}
    \label{eq:gen_planner_system}
    \begin{aligned}
        x^{\mathrm{H}}(k_{\mathrm{H}}+1) &= f^{\mathrm{H}}(x^{\mathrm{H}}(k_{\mathrm{H}}), r^{\mathrm{H}}(k_{\mathrm{H}})),
    \end{aligned}
\end{equation}
which is utilized by the higher-level controller. 
Moreover, $x^{\mathrm{H}}(k_{\mathrm{H}})\in\mathbb R^{n_{\mathrm{Hx}}}$ is the system state, 
and $k_{\mathrm{H}}\in\mathbb{Z}_{\geq 0}$ is the sampling index of the higher-level controller. 
The model can be an under-sampled or a reduced-order version of the original model \eqref{eq:constraints}. 
We specifically assume the existence of a mapping from the state $x$ to the state $x^\mathrm{H}$  represented as
\begin{equation}
    \label{eq:structure_state}
    x^\mathrm{H} := g(x).
\end{equation}
%
Moreover, the system \eqref{eq:gen_planner_system} is subject to the constraints $(x^{\mathrm{H}}(k_{\mathrm{H}}), r^{\mathrm{H}}(k_{\mathrm{H}})) \in \mathcal{Z}^{\mathrm{H}}.$ 
At the beginning of the mission, i.e., $k = k_{\mathrm{H}} = 0$, the higher-level controller solves the optimization problem
\begin{subequations}
    \label{eq:opt_planning}
    \begin{align}
        &\underset{x^{\mathrm{H}}_{ \cdot|k_{\mathrm{H}}}, 
        r^{\mathrm{H}}_{\cdot|k_{\mathrm{H}}}}{\mathrm{minimize}} & & J^{\mathrm{H}}(x^{\mathrm{H}}_{ \cdot|k_{\mathrm{H}}}, r^{\mathrm{H}}_{\cdot|k_{\mathrm{H}}})\\
        & \mathrm{ subject\;to} & &  x^{\mathrm{H}}_{ 0|k_{\mathrm{H}}} = x^{\mathrm{H}}(k_{\mathrm{H}}), \label{eq:planner_constraint_start}\\
        & & & x^{\mathrm{H}}_{ N_{\mathrm{H}}|k_{\mathrm{H}}} \in \mathcal{X}^{\mathrm{H}}_{\mathrm{f}}, \\
        & & & \mathrm{for\;} m = \{ 0,\ldots,N_{\mathrm{H}}-1\}: \nonumber \\
        & & & \quad x^{\mathrm{H}}_{ m+1|k_{\mathrm{H}}} = f^{\mathrm{H}}(x^{\mathrm{H}}_{m|k_{\mathrm{H}}}, r^{\mathrm{H}}_{ m|k_{\mathrm{H}}}),\\
        & & & \quad(x^{\mathrm{H}}_{m|k_{\mathrm{H}}}, r^{\mathrm{H}}_{m|k_{\mathrm{H}}}) \in \mathcal{Z}^{\mathrm{H}}. \label{eq:planner_constraint_end}
    \end{align}
\end{subequations}
It optimizes the reference sequence $r^{\mathrm{H}}_{\cdot|k_{\mathrm{H}}}$ and state sequence $x^{\mathrm{H}}_{\cdot|k_{\mathrm{H}}}$ over the mission horizon $N_{\mathrm{H}}$.
$x_{m|k_{\mathrm{H}}}$ denotes the predicted state for step $m+k_{\mathrm{H}}$ made at time step $k_{\mathrm{H}}$. 
Although we assume $k_{\mathrm{H}} = 0$ for the mission-based setup, we introduce this notation here for use in the receding-horizon context discussed in Section \ref{sec:receding_horizon}. 
The cost function to be minimized is represented by $J^{\mathrm{H}}(x^{\mathrm{H}}_{ \cdot|k_{\mathrm{H}}}, r^{\mathrm{H}}_{\cdot|k_{\mathrm{H}}})$, while $\mathcal{X}_{\mathrm{f}}^{\mathrm{H}}\subseteq\mathbb R^{n_{\mathrm{Hx}}}$ denotes a target or terminal set.
The optimal reference trajectory is denoted as $r^{\mathrm{H},*}_{\cdot|k_{\mathrm{H}}}$.
Due to the potential nonlinearity of the dynamics and the possibility of non-convex constraints and cost functions, the overall optimization problem is generally non-convex. This is often the case in various applications, including planning for autonomous driving~\cite{paden2016survey}.

The task of the lower-level controller is the execution of the plan by the higher-level controller. It uses the model~\eqref{eq:system} and receives the reference $r^{\mathrm{H},*}_{\cdot|k_{\mathrm{H}}}$ from the higher-level controller. 
The reference is held constant between sampling intervals, i.e., $r_{mN_{\mathrm{L}}+l|k_{\mathrm{H}}} := r^{\mathrm{H},*}_{m|k_{\mathrm{H}}},  \forall l \in \{0,\ldots,N_{\mathrm{L}}-1\}, m \in \{0, \ldots, N_{\mathrm{H}}-1\}$.
%
In addition to the constraints \eqref{eq:constraints}, the lower-level is subject to reference-dependent constraints of the form
\begin{align}\label{eq:reference_dependent_constraints}
    c_{\Delta \mathrm{x}}(x(k), r^{\mathrm{H},*}_{\cdot|\mathrm{k_{\mathrm{H}}}}) \leq 0, \forall k \in  \{0,\ldots,N_{\mathrm{H}}N_{\mathrm{L}}-1\}.
\end{align}
These constraints may include limitations on the deviation from a reference position in motion control applications, see Section \ref{sec:application_autonomous_driving}, or different operational modes for the lower-level controller, as selected by the higher-level, e.g., as discussed in~\cite{kogel2023safe} for autonomous navigation. 
While we focus on reference-dependent state constraints, this framework can also be extended to include reference-dependent input constraints, such as those found in process control~\cite{picasso2010mpc} or power systems control applications~\cite{berkel2013load}.

The goals of the presented design approach are twofold: first, to design a hierarchical control system that ensures safety through constraint satisfaction despite model differences between levels; and second, to achieve a modular, contract-based design.  
This design aims to minimize the shared information between the layers, ensuring that the model and optimization problem of the lower-level controller remain undisclosed to the higher-level controller and vice versa. 
Instead, a contract $h_{\mathrm{C}}$ is exchanged, allowing the higher level to check feasibility of the lower-level controller.
\section{Contract-based hierarchical control}\label{sec:contract_based_design}
We first introduce the design of the lower-level controller and state the corresponding feasibility problem. 
Subsequently we extend the higher-level controller to assess feasibility of the lower-level controller. 
Finally, we discuss the usage of the approximation of the predictive feasibility value function as contract $h_{\mathrm{C}}$ between the controller levels.
\subsection{Lower-level control based on soft-constrained \gls{MPC}}

The lower-level controller is formulated based on a soft-constrained \gls{MPC} approach. 
At every time step $k \in \{0,\ldots,N_{\mathrm{H}}N_{\mathrm{L}}-1\}$, it solves the optimization problem 
\begin{subequations}
	\label{eq:soft_constrained_mpc}
	\begin{align}
    &\underset{x_{\cdot|k}, u_{\cdot|k}, \xi_{\cdot|k}}{\mathrm{minimize}} & & J_{\mathrm{MPC}}(x_{\cdot|k}, u_{\cdot|k},        r^{\mathrm{H},*}_{\cdot|k_{\mathrm{H}}}) + w_{\xi}J_{\xi}(\xi_{\cdot|k}) \\
        & \mathrm{ subject\;to} & &  x_{0|k} = x(k), \label{eq:soft_constrained_mpc_constraint_start}\\
        & & & \text{for\;} l \in \{0,\ldots,N(k)-1\}:\\
        & & & \quad x_{l+1|k} = f( x_{l|k}, u_{l|k}), \label{eq:soft_constrained_mpc_model}\\
        & & & \quad c_{\mathrm{x}} (x_{l|k}) \leq \xi_{l|k}^{\mathrm{x}}, \label{eq:soft_constrained_mpc_constraint_1}\\
        & & & \quad c_{\Delta\mathrm{x}} (x_{l|k}, r^{\mathrm{H},*}_{\cdot|k_{\mathrm{H}}}) \leq  \xi_{l|k}^{\Delta\mathrm{x}},  \label{eq:soft_constrained_mpc_constraint_2}\\
        & & & \quad c_{\mathrm{u}} (u_{l|k}) \leq 0,   \label{eq:soft_constrained_mpc_constraint_3}\\
        & & & \quad \xi_{l|k} = [\xi_{l|k}^{\mathrm{x},\top}, \xi_{l|k}^{\Delta\mathrm{x},\top}]^\top \geq 0. \label{eq:soft_constrained_mpc_constraint_end}
	\end{align}
\end{subequations}
The sequences $x_{\cdot|k}$, $u_{\cdot|k}$, and $\xi_{\cdot|k}$ represent the state, input, and slack variables, respectively, over the prediction horizon. 
The optimization problem has a shrinking horizon with length $N(k) := N_{\mathrm{H}}N_{\mathrm{L}}-k$, as the optimization problem starts at time $k$ and ends at the end of the mission.
The constraints \eqref{eq:soft_constrained_mpc_constraint_1} and \eqref{eq:soft_constrained_mpc_constraint_2} are relaxed using slack variables $\xi_{l|k}\geq 0$. 
The cost function consists of two parts: An \gls{MPC} cost function and a feasibility cost function. The \gls{MPC} cost function is given by 
$J_{\mathrm{MPC}}(x_{\cdot|k}, u_{\cdot|k},        r^{\mathrm{H}}_{\cdot|k_{\mathrm{H}}})
    := \sum_{l = 0}^{N(k)-1} \ell(x_{l|k}, u_{l|k}, r^{\mathrm{H}}_{\cdot|k_{\mathrm{H}}})$,
%
where $\ell(x, u, r)$ denotes the stage cost. 
The feasibility cost function penalizes the slack variables
$\xi_{l|k}$ and is defined as
%
$J_{\xi}(\xi_{\cdot|k}) := 
\sum_{l = 0}^{N(k)-1} \|\xi_{l|k}\|_1,$
%
where $\|v\|_1$ denotes the one norm of a vector $v$.

By selecting a sufficiently large weight factor $w_{\xi}$, the soft-constrained \gls{MPC} problem yields an optimal input sequence identical to that of its hard-constrained counterpart, provided the latter is feasible, resulting in $\xi_{l|k}^{*} = 0$ for all $l \in \{0,\ldots,N(k)-1\}$, see, e.g.,~\cite{kerrigan2000soft} for further details. 
Additionally, the soft-constrained \gls{MPC} problem remains feasible for states where the hard-constrained \gls{MPC} problem is infeasible, due to the relaxation of state constraints, leading to $\xi_{l|k}^{*} > 0$ for some $l \in \{0,\ldots,N(k)-1\}$.
Given the state $x(k)$ and the reference $r^{\mathrm{H},*}_{\cdot|k_{\mathrm{H}}}$, the soft-constrained \gls{MPC} problem \eqref{eq:soft_constrained_mpc} is solved 
online to obtain the optimal input sequence $u_{\cdot|k}^{*}$ and the corresponding slack variable sequence $\xi_{\cdot|k}^{*}$. The first input 
$u_{0|k}^{*}$ is applied to the system~\eqref{eq:system}.

As detailed in the next section, the soft-constrained formulation allows to identify infeasible state and reference combinations through the optimal slack variables, which serves as a basis for the desired contract $h_{\mathrm{C}}$.

For further analysis, we introduce the set of input sequences
%
    $\mathcal{U}_{k}(x(k), r_{\cdot|k_{\mathrm{H}}}^{\mathrm{H}}):= \left\{u_{\cdot|k}| \eqref{eq:soft_constrained_mpc_constraint_start}-\eqref{eq:soft_constrained_mpc_constraint_end} \wedge \xi_{\cdot|k} = 0\right\},$
for which the constraints are fulfilled without constraint relaxation at time $k$.
\subsection{Feasibility-aware higher-level control}
To establish the contract, we define the feasibility problem
\begin{subequations}
    \label{eq:feasibility_problem}
    \begin{align}
    h^*(x(k), r^{\mathrm{H}}_{\cdot|k_{\mathrm{H}}}) := &\underset{{
 x_{\cdot|k}, u_{\cdot|k}, \xi_{\cdot|k}}}{\mathrm{min}} & & J_{\xi}(\xi_{\cdot|k})\\
    & \mathrm{ subject\;to} & & \eqref{eq:soft_constrained_mpc_constraint_start} - \eqref{eq:soft_constrained_mpc_constraint_end}.
    \end{align}
\end{subequations}
corresponding to \eqref{eq:soft_constrained_mpc}. 
The feasibility problem checks whether for a given output trajectory $r^{\mathrm{H},*}_{\cdot|k_{\mathrm{H}}}$ and state $x(k)$ there exists an input sequence for the lower-level controller which complies with the constraints \eqref{eq:constraints} and \eqref{eq:reference_dependent_constraints} assuming an evolution according to the model \eqref{eq:system} over the complete mission horizon, i.e., $N(k) := N_{\mathrm{H}}N_{\mathrm{L}}$. 
We call $h^*(x(k), r^{\mathrm{H},*}_{\cdot|k_{\mathrm{H}}})$ the predictive feasibility value function. 
In case constraints are violated, the predictive feasibility value function $h^*(x(k), r^{\mathrm{H},*}_{\cdot|k_{\mathrm{H}}})$ is positive. 
When all constraints can be satisfied, the value function $h^*(x(k), r^{\mathrm{H},*}_{\cdot|k_{\mathrm{H}}}) = 0$. Note that $h^*(x(k), r^{\mathrm{H},*}_{\cdot|k_{\mathrm{H}}}) = 0$ implies that optimization problem \eqref{eq:soft_constrained_mpc} is feasible with $\xi_{l|k}^{*} = 0$ for all $l \in \{0,\ldots,N_{\mathrm{H}}N_{\mathrm{L}}\}$. 

We leverage the predictive feasibility value function as a foundation for a contract that enables the higher-level controller to assess the feasibility of trajectories for the lower-level controller. 
To achieve this, we incorporate the value function into the cost of the higher-level optimization problem \eqref{eq:opt_planning}, resulting in
\begin{subequations}
    \label{eq:opt_planning_contract_based}
    \begin{align}
        &\underset{x^{\mathrm{H}}_{ \cdot|k_{\mathrm{H}}},  r^{\mathrm{H}}_{\cdot|k_{\mathrm{H}}}}{\mathrm{minimize}} & & J^{\mathrm{H}}(x^{\mathrm{H}}_{ \cdot|k_{\mathrm{H}}}, r^{\mathrm{H}}_{\cdot|k_{\mathrm{H}}}) + 
        w_{\mathrm{h}}h^*(x(k), r^{\mathrm{H}}_{\cdot|k_{\mathrm{H}}})\\
        & \mathrm{ subject\;to} & &  \eqref{eq:planner_constraint_start} - \eqref{eq:planner_constraint_end}.
    \end{align}
\end{subequations}
%
The weighting factor $w_{\mathrm{h}} \in \mathbb{R}_{\geq 0}$ can be used to trade-off the violation of the constraint in the lower-level and the minimization of costs of $J^{\mathrm{H}}(x^{\mathrm{H}}_{ \cdot|k_{\mathrm{H}}}, r^{\mathrm{H}}_{\cdot|k_{\mathrm{H}}})$. 

For further analysis, we introduce the set 
$\mathcal{H} := \left\{(x(k), r^{\mathrm{H}}_{\cdot|k_{\mathrm{H}}}) | h^*(x(k), r^{\mathrm{H}}_{\cdot|k_{\mathrm{H}}}) = 0\right\},$
which is the set of states and references for which there exists a solution such that the predictive feasibility value function is zero.
Next, we summarize the main result for the mission-based setup where we state constrain satisfaction of~\eqref{eq:constraints} through recursive feasibility of the lower-level controller.
\begin{theorem}
    Consider the higher-level optimization problem~\eqref{eq:opt_planning_contract_based} at time $k = k_{\mathrm{H}} = 0$. 
    If the     optimal solution $r^{\mathrm{H},*}_{\cdot|0}$ is such that $(x(0),r^{\mathrm{H},*}_{\cdot|0}) \in  \mathcal{H}$, 
    then there exist feasible input sequences for the lower-level controller with slacks equal to zero for the entire mission.
\end{theorem}
\begin{proof}
    From $(x(0),r^{\mathrm{H},*}_{\cdot|0}) \in  \mathcal{H}$ follows that there exists an input sequence $u_{\cdot|0}^* \in \mathcal{U}_{0}(x(0), r_{\cdot|0}^{\mathrm{H,*}})$. Input sequences $u_{\cdot|k} \in \mathcal{U}_{k}(x(k), r_{\cdot|0}^{\mathrm{H,*}})$ for $k =  1, \ldots,  N_{\mathrm{H}}N_{\mathrm{L}}-1$ can be constructed according to $u_{\cdot|k +1} = \left\{ u^*_{l|k}\right\}_{l = 1, \ldots,  N_{\mathrm{H}}N_{\mathrm{L}}-1}$ completing the proof.
\end{proof}

\subsection{Contract design using value function approximation}
\label{subsec:value_function_approx}
In the following discussion, we propose using an explicit function approximation, denoted as $h_{\mathrm{C}}(x(k), r^{\mathrm{H},*}_{\cdot|k_{\mathrm{H}}})$, for the implicit predictive feasibility value function $h^*(x(k), r^{\mathrm{H},*}_{\cdot|k_{\mathrm{H}}})$, such that $h_{\mathrm{C}}(x(k), r^{\mathrm{H},*}_{\cdot|k_{\mathrm{H}}}) \approx h^*(x(k), r^{\mathrm{H},*}_{\cdot|k_{\mathrm{H}}})$. 
This approximation, which may take the form of an \gls{NN} or \gls{LUT}, serves as the contract between the lower-level and higher-level controllers. 
The first reason for this approach is that directly incorporating the value function $h^*(x(k), r^{\mathrm{H},*}_{\cdot|k_{\mathrm{H}}})$ in the cost function, renders \eqref{eq:opt_planning_contract_based} a nested optimization problem, which is computationally demanding to evaluate in real-time. 
To enable efficient implementation, the value function can be approximated using an explicit approximator. 
For its parameterization, the high-level planner has to provide trajectories $r^{\mathrm{H}}$ to the lower-level controller. The lower-level controller then uses these trajectories $r^{\mathrm{H}}$ together with sampled states $x$ and solves the feasibility problem to obtain the predictive feasibility value function $h^*$. 
All this data is used to generate a training data set $(x, r^{\mathrm{H}},h^*)_{i
=
1,\ldots, N_
{\mathrm{data}}}$
which can be used for parameterizing the approximator using supervised learning. 
The trained model $h_{\mathrm{C}}$ is then transferred to the higher-level controller.
When the higher-level controller is executed online, the function approximation can be employed to evaluate the feasibility of the generated trajectory. 
In cases where the optimization scheme of the higher-level controller is solved using sampling-based methods, as, e.g., common in planning in autonomous driving, see~\cite{paden2016survey}, due to the non-convex nature of the optimization problem, the feasibility of references for the lower-level controller can be easily assessed through evaluation of the function approximation. 
If the optimization scheme works in a gradient-based manner, the gradients of the function approximation can be analytically computed. 

The second reason for utilizing the value function approximation as the contract $h_{\mathrm{C}}$ is that it provides the higher-level controller with an abstract representation of the optimization problem's solution, rather than the complete optimization problem that includes intricate details about costs, models, and constraints. 
This approach simplifies the integration of lower-level controllers and addresses potential limitations on knowledge exchange regarding models due to intellectual property concerns, all while minimizing the need for detailed understanding of lower-level specifics.

\section{Extension to receding-horizon control}
\label{sec:receding_horizon}
In this section, we extend our analysis beyond mission-based scenarios to consider a higher-level controller that operates in a receding-horizon fashion allowing for continuous operation of the overall controller architecture.

In this set-up, the higher-level controller problem \eqref{eq:opt_planning_contract_based} is solved every time step $k \mod N_{\mathrm{L}} = 0$, that is, $k = nk_{\mathrm{H}}$ with $n \in \mathbb{Z}_{\geq 0}$ and the terminal set is chosen as a steady state manifold, i.e.,
\begin{align}
\label{eq:planner_terminal_eq_constraint}
\mathcal{X}^{\mathrm{H}}_{\mathrm{f}} := \left\{x^{\mathrm{H}}_{\mathrm{s}}| x^{\mathrm{H}}_{\mathrm{s}} = f^{\mathrm{H}}(x^{\mathrm{H}}_{\mathrm{s}}, r^{\mathrm{H}}_{\mathrm{s}}), 
(x^{\mathrm{H}}_{\mathrm{s}}, r^{\mathrm{H}}_{\mathrm{s}}) \in \mathcal{Z}^{\mathrm{H}}\right\}.
\end{align}
Moreover, we make the following assumptions on the steady states which ensures that there exists a feasible steady state for the lower-level controller model for feasible steady states of the model for the higher-level controller.
\begin{assumption}
\label{as:steady_state}
For any steady state pair $(x^{\mathrm{H}}_{\mathrm{s}},r^{\mathrm{H}}_{\mathrm{s}})$ with $x^{\mathrm{H}}_{\mathrm{s}} \in \mathcal{X}^{\mathrm{H}}_{\mathrm{f}}$, there exists a steady state pair $(x_{\mathrm{s}},u_{\mathrm{s}})$ such that $x_{\mathrm{s}} = f(x_{\mathrm{s}}, u_{\mathrm{s}})$ with $x_{\mathrm{s}} \in \mathcal{X}$, $u_{\mathrm{s}} \in \mathcal{U}$, and $c_{\Delta \mathrm{x}}(x_{\mathrm{s}}, r^{\mathrm{H}}_{\mathrm{s}}) \leq 0$ holds.
\end{assumption}
%

Furthermore, the feasibility problem is adapted to
\begin{subequations}
    \label{eq:feasibility_problem_extended}
    \begin{align}
    &\underset{{ x_{\cdot|k},  u_{\cdot|k}, \xi_{\cdot|k}}}{\mathrm{minimize}} & & J_{\xi}(\xi_{\cdot|k})\\
    & \mathrm{ subject\;to} & & \eqref{eq:soft_constrained_mpc_constraint_start} - \eqref{eq:soft_constrained_mpc_constraint_end}, \label{eq:soft_constrained_mpc_extended_constraint_1}\\
    & & & \mathrm{for}\;m\in\{1,\ldots,N_{\mathrm{H}}\}: \nonumber \\ 
    & & & \xi^{\mathrm{g}}_{m|k} \leq g(x_{mN_{\mathrm{L}}|k}) - x^{\mathrm{H}}_{m|k_{\mathrm{H}}} \leq \xi^{\mathrm{g}}_{m|k} ,\label{eq:soft_constrained_mpc_extended_constraint_end}
    \end{align}
\end{subequations}
where compared to \eqref{eq:feasibility_problem} constraints \eqref{eq:soft_constrained_mpc_extended_constraint_end} are added to ensure consistency between higher and lower-level controller at the sampling times of the higher-level controller.
$\xi^{\mathrm{g}}_{m|k}$ serves as a slack variable for these constraints and is included in the slack sequence $\xi_{\cdot|k}$.
The feasibility problem is again solved over the complete horizon, i.e., with $N := N_{\mathrm{H}}N_{\mathrm{L}}$. 
From a computational perspective, this is not an issue, as the problem is approximated offline and online the efficient explicit function approximation is evaluated.

The lower-level controller operates on a cyclic horizon $N(k): = N_{\mathrm{L}}-k \mod N_{\mathrm{L}}$ which is generally shorter than in the mission-based case. 
It is evaluated every sampling instant $k$ and solves the optimization problem
\begin{subequations}
	\label{eq:soft_constrained_mpc_extended}
	\begin{align}
    &\underset{x_{\cdot|k}, u_{\cdot|k}, \xi_{\cdot|k}}{\mathrm{minimize}} & & J_{\mathrm{MPC}}(x_{\cdot|k}, u_{\cdot|k}, r^{\mathrm{H}}_{\cdot|k_{\mathrm{H}}}) + w_{\xi}J_{\xi}(\xi_{\cdot|k}) \\
        & \mathrm{ subject\;to} & &  \eqref{eq:soft_constrained_mpc_constraint_start} - \eqref{eq:soft_constrained_mpc_constraint_end},\\
        & & & \xi^{\mathrm{g}}_{k} \leq g(x_{N(k),k}) - x^{\mathrm{H}}_{1,k_{\mathrm{H}}} \leq \xi^{\mathrm{g}}_{k}, \label{eq:mpc_terminal_eq_constraints}
	\end{align}
\end{subequations}
where compared to \eqref{eq:soft_constrained_mpc} constraint \eqref{eq:mpc_terminal_eq_constraints} is added to ensure consistency between higher and lower-level controller at the next sampling time of the higher-level controller. 
$\xi^{\mathrm{g}}_{k}$ serves as a slack variable for this constraint and is included in the slack sequence $\xi_{\cdot|k}$.

In this setting, recursive feasibility of both the lower- and higher-level controllers ensures that the constraints in \eqref{eq:constraints} are satisfied at all times.
\begin{theorem}
    Consider the higher-level optimization problem~\eqref{eq:opt_planning_contract_based} with terminal constraint \eqref{eq:planner_terminal_eq_constraint} at time $k = k_{\mathrm{H}}N_{\mathrm{L}}$ with $k_{\mathrm{H}}\in \mathbb{Z}_{\geq0}$. If the 
    optimal solution $(x(k),r^{\mathrm{H},*}_{\cdot|k_{\mathrm{H}}}) \in  \mathcal{H}$, then 
    there exists a solution for the lower-level controller with slack variables equal to zero for all $k =  k_{\mathrm{H}}N_{\mathrm{L}}, \ldots, (k_{\mathrm{H}}+1)N_{\mathrm{L}}-1$. Moreover, there exists a feasible solution for the higher-level controller at time $k_{\mathrm{H}}+1$. 
\end{theorem}
\begin{proof}
    The first statement of the theorem  directly follows from Theorem 1.
    
    Next, we consider the second statement.  
    First, we consider feasibility of \eqref{eq:opt_planning_contract_based} at time $k_{\mathrm{H}}+1$. 
    From the first part of the proof follows that $u_{\cdot|(k_{\mathrm{H}}+1)N_{\mathrm{L}}-1} \in \mathcal{U}_{(k_{\mathrm{H}}+1)N_{\mathrm{L}}-1}$. 
    Together with
    constraint \eqref{eq:mpc_terminal_eq_constraints} and $\xi^{\mathrm{g}}_{(k_{\mathrm{H}}+1)N_{\mathrm{L}}-1} = 0$, it follows that $x^{\mathrm{H}}(k_{\mathrm{H}}+1) = x^{\mathrm{H}}_{1|k_{\mathrm{H}}}$. 
    Considering the terminal constraint design \eqref{eq:planner_terminal_eq_constraint} and following standard \gls{MPC} arguments, the shifted sequence $r^{\mathrm{H}}_{\cdot|k_{\mathrm{H}}+1} := \left\{r^{\mathrm{H},*}_{1|k_{\mathrm{H}}}, \ldots, r^{\mathrm{H},*}_{N_{\mathrm{H}}|k_{\mathrm{H}}}, r^{\mathrm{H}}_{\mathrm{s}}\right\}$ satisfies \eqref{eq:planner_constraint_start} - \eqref{eq:planner_constraint_end}.
    
    Next, feasibility of the feasibility problem  \eqref{eq:feasibility_problem_extended} is checked.
    Again following standard MPC arguments, the candidate solution 
    $u_{\cdot|(k_{\mathrm{H}}+1)N_{\mathrm{L}}} := \{u_{k_{\mathrm{H}}N_{\mathrm{L}}|k_{\mathrm{H}}N_{\mathrm{L}}}, \ldots,\underbrace{u_{\mathrm{s}}^{*},\ldots,u_{\mathrm{s}}^{*}}_{N_{\mathrm
        L} steps}\}$,
    satisfies the constraints \eqref{eq:soft_constrained_mpc_extended_constraint_1}-\eqref{eq:soft_constrained_mpc_extended_constraint_end} with $\xi_{l|(k_{\mathrm{H}}+1)N_{\mathrm{L}}} = 0$ for $l = 0,\ldots,N$ and $\xi^{\mathrm{g}}_{k_{\mathrm{H}}|(k_{\mathrm{H}}+1)N_{\mathrm{L}}} = 0$ for $m=1,\ldots,N_\mathrm{H}$ due to \eqref{eq:soft_constrained_mpc_extended_constraint_end} at time $k = k_{\mathrm{H}}N_{\mathrm{L}}$ and Assumption \ref{as:steady_state}. 
    This implies that $(x(k+N_{\mathrm{L}}),r^{\mathrm{H}}_{\cdot|k_{\mathrm{H}}+1}) \in  \mathcal{H}$, completing the proof.
\end{proof}
\section{Application to autonomous driving}
\label{sec:application_autonomous_driving}
We illustrate our proposed approach within the context of contract-based planning and motion control for autonomous driving in a mission-based obstacle avoidance scenario. 
To capture the vehicle dynamics, we employ a nonlinear single-track model of the form
\begin{equation}
    \begin{aligned}
    \label{eq:nonlinear_single_track_model}
    \dot{p}_{\mathrm{x}} &= v\cos{(\psi)}, \quad \dot{p}_{\mathrm{y}} = v\sin{(\psi)}, \quad \dot{\psi} = \dot{\psi}, \quad \dot{v} = a,\\
    \ddot{\psi} &= \frac{1}{I_{\mathrm{z}}}[F_{\mathrm{fy}}\cos{(\delta)}l_{\mathrm{f}}-F_{\mathrm{ry}}l_{\mathrm{r}}],\\
    \dot{\beta} &= \frac{1}{mv}[F_{\mathrm{fy}}\cos{(\beta-\delta)}+F_{\mathrm{ry}}\cos{(\beta)}] - \dot{\psi}.
\end{aligned}
\end{equation}
where $p_{\mathrm{x}}$ is the position along the horizontal axis, $p_{\mathrm{y}}$ is the position along the vertical axis, $\psi$ is the orientation angle, $v$ is the velocity, $\beta$ is the side slip angle. 
The inputs are the steering angle $\delta$ and the acceleration $a$, i.e., $u = [\delta, a]^\top$. 
$F_{\mathrm{fy}}(v, \beta,\dot{\psi},\delta)$ is the lateral force on the front tires, $F_{\mathrm{ry}}(v, \beta,\dot{\psi},\delta)$ is the lateral force on the rear tires. 
Both are nonlinear functions of the state. 
$I_{\mathrm{z}}$ is the moment of inertia about the vertical axis, $l_{\mathrm{f}}$ is the distance to the front axle, $l_{\mathrm{r}}$ is the distance to the rear axle, and $m$ is the mass of the vehicle.
For additional details on the model and its parameters, please refer to~\cite{Althoff2017a}. 
The discrete-time model is derived using the Runge-Kutta 4th-order method with sampling time $T_{\mathrm{L}} =  \SI{20}{\milli \second}$.

The planner operates under the assumption of constant orientation and speed. It utilizes the simplified model
\begin{equation}
    \label{eq:planner_model}
    \begin{aligned}
        \dot{p}_{\mathrm{x}} &= v\cos{\psi},\quad
        \dot{p}_{\mathrm{y}} &= v\sin{\psi},
    \end{aligned}
\end{equation}
with state $x^{\mathrm{H}} = [p_{\mathrm{x}}, p_{\mathrm{y}}]^\top$ and input $r^{\mathrm{H}} = [\psi, v]^\top$. 
Note that the states in the model are also part of the state of model~\eqref{eq:nonlinear_single_track_model} and hence the mapping \eqref{eq:structure_state} is given as $g(x):=S^{\mathrm{H}}x$ where $S^{\mathrm{H}}$ is an appropriate selection matrix.  
The discrete-time model is again obtained through the RK4 method, with sampling time $T_{\mathrm{L}}$ and concatenated over $N_{\mathrm{L}} = 50$ steps to yield a higher-level model with a sampling time $T_{\mathrm{H}}$.
The prediction horizon is set to $N_{\mathrm{H}} = 2$. 
The planner incorporates non-convex constraints for obstacle avoidance, expressed as $\|x^{\mathrm{H}}-x_{\mathrm{obstacle}}\| > d_{\mathrm{obstacle}} + d_{\mathrm{max}}$. 
Here, $x_{\mathrm{obstacle}} = [15, 0]^\top$ denotes the obstacle's position, $d_{\mathrm{obstacle}} = 3$ is the obstacle's radius, and $d_{\mathrm{max}} = 2$ signifies the maximum allowed deviation of the controller from the reference provided by the planner. 
The planner's cost function is 
$J^{\mathrm{H}}(x^{\mathrm{H}}_{ \cdot|k_{\mathrm{H}}}, r^{\mathrm{H}}_{\cdot|k_{\mathrm{H}}}) :=  \sum_{m=0}^{N_{\mathrm{H}}} \|x^{\mathrm{H}}_{m|k_{\mathrm{H}}}-x^{\mathrm{target}}\|_{Q^{\mathrm{H}}}^2 + (v^{\mathrm{H}}_{m|k_{\mathrm{H}}}-v^{\mathrm{target}})^2,$
where $\|v\|_{M}:=v^\top M v$, $Q^{\mathrm{H}} = 10 \cdot I_{2\times2}$ is a weight matrix, $x^{\mathrm{target}} = [50, 10]^\top$ is a desired target position, and $v^{\mathrm{target}} = \SI{8.33}{m/s}$ is the the target speed. 
The planning problem \eqref{eq:opt_planning_contract_based} is solved using a sampling-based method. 
For the first prediction step, yaw angle and velocity are Gaussian-sampled. 
The second step's yaw angle and velocity are then computed to reach the target.

For the lower-level controller, we use the nonlinear single-track model \eqref{eq:nonlinear_single_track_model}. We consider box constraints on the velocity $\ubar{v} \leq v \leq \bar{v}$ as well as on the inputs $\ubar{\delta} \leq \delta \leq \bar{\delta}$ and $\ubar{a} \leq a \leq \bar{a}$.
The constraints \eqref{eq:reference_dependent_constraints} are derived from the constant speed and yaw angle, utilizing planner model \eqref{eq:planner_model} discretized with $T_{\mathrm{L}}$ to forward propagate to obtain $p^{\mathrm{ref}}_{\mathrm{x},l|\mathrm{k}_{\mathrm{H}}}$ and $p^{\mathrm{ref}}_{y,l|\mathrm{k}_{\mathrm{H}}}$ with which we build the constraint on the path deviation
\begin{align}
    \label{eq:relative_position_constraints}
    (p_{\mathrm{x},l|\mathrm{k}}- p^{\mathrm{ref}}_{\mathrm{x},l|\mathrm{k}_{\mathrm{H}}})^2 +  
    (p_{\mathrm{y},l|\mathrm{k}}-p^{\mathrm{ref}}_{\mathrm{y},l|\mathrm{k}_{\mathrm{H}}})^2 \leq  d_{\mathrm{max}}^2,\\
    \mathrm{for}\; k_{\mathrm{H}} = 0, \; k = 0,\ldots,N_{\mathrm{L}}, \; l = 0, \ldots, N_{\mathrm{H}}N_{\mathrm{L}}-k. \nonumber
\end{align}
We further use $p^{\mathrm{ref}}_{\mathrm{x},l|\mathrm{k}_{\mathrm{H}}}$ and $p^{\mathrm{ref}}_{y,l|\mathrm{k}_{\mathrm{H}}}$ in the stage cost which is selected as $\ell(x_{\cdot|k}, u_{\cdot|k}, r^{\mathrm{H}}_{\cdot|k_{\mathrm{H}}}):= \|x_{l|\mathrm{k}}- x^{\mathrm{ref}}_{l|\mathrm{k}_{\mathrm{H}}}\|^2_Q +\|u_{l|\mathrm{k}}\|^2_R$ 
with reference $x^{\mathrm{ref}}_{l|\mathrm{k}_{\mathrm{H}}} := [p^{\mathrm{ref}}_{\mathrm{x},l|\mathrm{k}_{\mathrm{H}}}, p^{\mathrm{ref}}_{\mathrm{y},l|\mathrm{k}_{\mathrm{H}}}, \psi^{\mathrm{ref}}_{l|\mathrm{k}_{\mathrm{H}}}, 0, 0 , 0]^\top$, $Q = \mathrm{diag}(10, 10, 1, 1, 0,0)$, and $R = \mathrm{diag}(0.1, 1)$.

We approximate the predictive feasibility value function $h^*(x(k), r_{\cdot|k_{\mathrm{H}}})$ by utilizing an \gls{LUT} with linear interpolation. 
For simplicity, assuming a fixed initial vehicle state, the \gls{LUT}'s input is defined by the velocity and angle of the planner's first prediction step, i.e., $r_{0|k_{\mathrm{H}}} \in \mathbb{R}^2$, effectively approximating $h^*(r_{0|k_{\mathrm{H}}})$. 
The \gls{LUT} itself contains 450 tuples of the form $(r_{0|k_{\mathrm{H}}}, h^*(r_{0|k_{\mathrm{H}}}))$, where the values of $r_{0|k_{\mathrm{H}}}$ are generated using Sobol sampling and $h^*$ pre-computed by solving the feasibility optimization problem.

We analyze an obstacle avoidance scenario (Figure \ref{fig:trajectory_planner}), showing sampled trajectories and their associated feasibility values. 
Starting at 
$x(0) = [0,0,-\frac{\pi}{4}, 8.33, 0,0]^\top$, trajectories above the obstacle are infeasible, while some below are feasible. 
Figure \ref{fig:cl_sim_with_collision} shows constraint violations without our contract-based design, contrasting with Figure \ref{fig:cl_sim_wo_collision}, where the contract-based design safely navigates the obstacle. 
Comparing approximated and true feasibility, all trajectories predicted feasible by the approximation are genuinely feasible. 
However, 26.9\% labeled infeasible by the approximation are actually feasible, indicating conservatism of the \gls{LUT} approximation. 
This motivates future work on improved approximations, such as NN-based methods as in \cite{chatzikiriakos2024learning}.
\begin{figure}[h]
    \centering
    \includegraphics[width=0.48\textwidth]{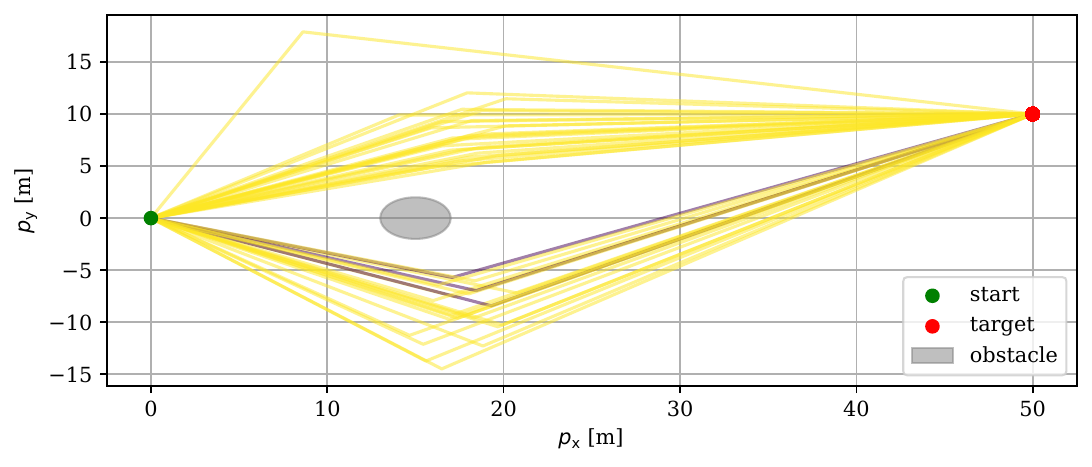}
    \caption{The figure shows the considered scenario with obstacle as gray box, green dot as initial, and red dot as target position. A total of 100 samples are generated by the sampling-based planning algorithm. Only those satisfying the planner's constraints in \eqref{eq:opt_planning} are displayed. Yellow trajectories, where $h^*(x(k),r^{\mathrm{H}}_{\cdot|k_{\mathrm{H}}}) > 0$, lead to a violation of the constraints of the controller. 
    The purple trajectories, where $h_{\mathrm{C}}(x(k),r^{\mathrm{H}}_{\cdot|k_{\mathrm{H}}}) = 0$, satisfy the constraints of the controller.
    }
    \label{fig:trajectory_planner}
\end{figure}
\begin{figure}[h]
    \centering
    \includegraphics[width=0.48\textwidth]{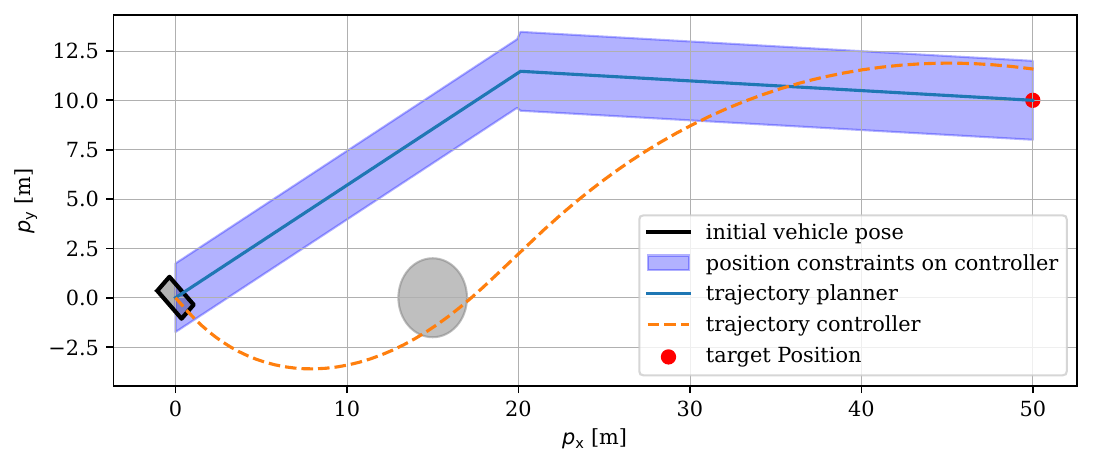}
    \caption{In this figure, we first consider the reference trajectory that a planner without contract-based design would have chosen. The controller trajectory clearly shows that this plan violates the position constraints on the controller (shown as blue area) given by \eqref{eq:relative_position_constraints}, leading to a direct collision with the obstacle. This highlights the necessity of the contract-based approach, as for the proposed reference $h_{\mathrm{C}}(x(k),r^{\mathrm{H}}_{\cdot|k_{\mathrm{H}}}) > 0$ holds and hence this reference could have been readily handled by such a design.}
    \label{fig:cl_sim_with_collision}
\end{figure}
\begin{figure}[h]
    \centering
    \includegraphics[width=0.48\textwidth]{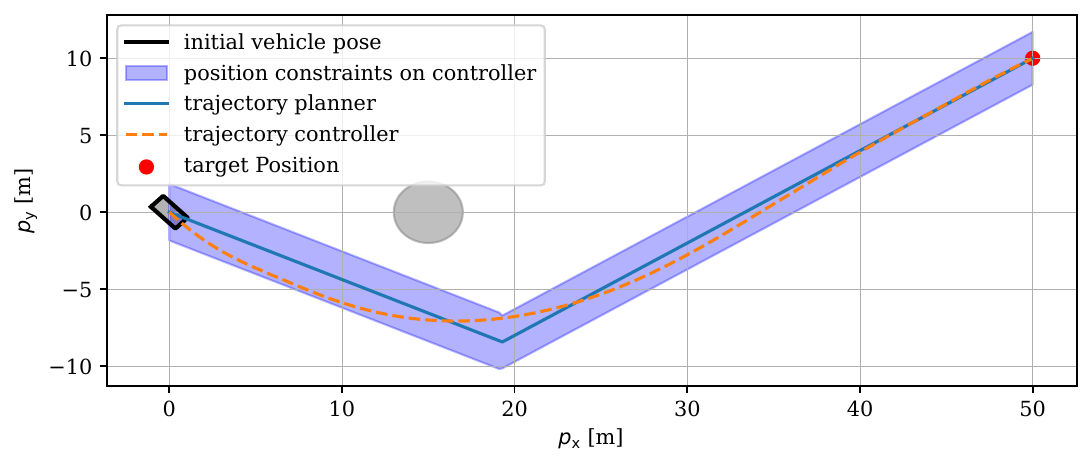}
    \caption{ In this figure, we analyze the reference trajectory generated by a planner incorporating the contract-based design. With the proposed reference from the higher-level controller, $h_{\mathrm{C}}(x(k),r^{\mathrm{H}}_{\cdot|k_{\mathrm{H}}}) = 0$ holds and hence the resulting trajectory is feasible for the lower-level controller. 
    The controller trajectory confirms this plan avoids violating the constraints in \eqref{eq:relative_position_constraints}.}
    \label{fig:cl_sim_wo_collision}
\end{figure}

\section{Conclusion}\label{sec:conclusion}
This paper introduced a contract-based hierarchical control strategy for modularization and safety. 
Using an approximate predictive feasibility value function as a contract, the higher-level controller efficiently evaluates reference trajectory feasibility without needing detailed lower-level model, constraint, or cost function knowledge. 
We validated our method through an autonomous driving case study with a hierarchical planner and motion controller.
%
\bibliographystyle{IEEEtran}
\bibliography{bib}

\end{document}